\documentclass[12pt,english]{article}
\usepackage{amsmath}
\usepackage{amsfonts}
\usepackage{amssymb}
\usepackage{amsthm}

\usepackage[left=25mm,top=26mm,right=26mm,bottom=30mm]{geometry}

\usepackage{enumerate}

\newtheorem{theorem}{Theorem}[section]

\newtheorem{lemma}{Lemma}[section]

\def\tire{\thinspace--\thinspace}

\title{Long Time Behavior of Infinite Harmonic Chain with $l_{2}$ Initial
  Conditions}

\author{A.A.~Lykov\thanks{Mechanics and Mathematics Faculty, Lomonosov Moscow
  State University, Leninskie Gory~1, Moscow, 119991, Russia} \and M.V.~Melikian \footnotemark[1] }

\begin{document}
\maketitle

\begin{abstract}
We consider infinite harmonic chain with $l_{2}$-initial conditions
and deterministic dynamics (no probability at all). Main results concern
the question when the solution will be uniformly bounded in time and
space in the $l_{\infty}$-norm. 
\end{abstract}

\allowdisplaybreaks 
\section{Introduction}

Consider a countable system of point particles with unit masses on
$\mathbb{R}$ with coordinates $\{x_{k}\}_{k\in\mathbb{Z}}$ and velocities
$\{v_{k}\}_{k\in\mathbb{Z}}$. We define formal energy (hamiltonian)
by the following formula: 
\[
H=\sum_{k\in\mathbb{Z}}\frac{v_{k}^{2}}{2}+\frac{\omega_{0}^{2}}{2}\sum_{k\in\mathbb{Z}}(x_{k}(t)-ka)^{2}+\frac{\omega_{1}^{2}}{2}\sum_{k\in\mathbb{Z}}(x_{k}(t)-x_{k-1}(t)-a)^{2},
\]
with parameters $a>0,\ \omega_{1}>0,\ \omega_{0}\geqslant0$. Particle
dynamics is 
defined by the infinite system of ODE: 
\begin{align}
  \ddot{x}_{k}(t)&=-\frac{\partial{H}}{\partial{x_{k}}}=-\omega_{0}^{2}(x_{k}(t)-ka)+\omega_{1}^{2}(x_{k+1}(t)-x_{k}(t)-a) \notag \\
&\quad {}   -\omega_{1}^{2}(x_{k}(t)-x_{k-1}(t)-a),\quad k\in\mathbb{Z}\label{mainEqX}
\end{align}
with initial conditions $x_{k}(0),v_{k}(0)$. The equilibrium state
(minimum of the energy) is 
\[
x_{k}=ka,\quad v_{k}=0,\quad k\in\mathbb{Z}.
\]
This means that if the initial conditions are in the equilibrium state
then the system will not evolve, i.e. $x_{k}(t)=ka,\ v_{k}(t)=0$
for all $t\geqslant0$. Let us introduce deviation variables:
\[
q_{k}(t)=x_{k}-ka,\quad p_{k}(t)=\dot{q}_{k}(t)=v_{k}(t).
\]
Our main assumption is $q(0)=\{q_{k}(0)\}_{k\in\mathbb{Z}}\in l_{2}(\mathbb{Z}),\ p(0)=\{p_{k}(0)\}_{k\in\mathbb{Z}}\in l_{2}(\mathbb{Z})$.
In the present article we study the long time behavior of $q_{k}(t)$
depending on initial conditions and parameters $a,\omega_{0},\omega_{1}$.
Namely, we are interested in the uniform boundedness (in $k$ and
$t$), the order of growth (in $t$) and exact asymptotic behavior
of $q_{k}(t)$.

It is easy to see that $q_{k}(t)$ satisfies the following system
of ODE: 
\begin{equation}
\ddot{q}_{k}=-\omega_{0}^{2}q_{k}+\omega_{1}^{2}(q_{k+1}-q_{k})-\omega_{1}^{2}(q_{k}-q_{k-1}),\quad k\in\mathbb{Z}.\label{mainEqForQ}
\end{equation}
The system of coupled harmonic oscillators (\ref{mainEqForQ}) and
its generalizations is a classical object in mathematical physics.
The existence of solution and its ergodic properties were 
studied in \cite{LanfordLebowitz}. There has been an 
 extensive research of convergence to equilibrium for infinite harmonic
chain coupled with a heat bath \cite{Bogolyubov,DudKomechSpohn,SpohnLebowitz,BPT}.
The property of uniform boundedness (by time $t$ and index $k)$
is crucial in some applications. For instance, uniform boundedness
in finite harmonic chain allows to derive Euler equation and Chaplygin
gas without any stochastics (see \cite{LM1}). Uniform boundedness
of a one-side non-symmetrical harmonic chain play important role in
some traffic flow models \cite{LMM}. We should 
note some physical papers \cite{Hemmen,Fox,FlorencioLee}. The most
closely related works to ours are \cite{Dud1,Dud2}, where the author
studied weighted $l_{2}$ norms of infinite harmonic chains, whereas
our main interest is a max-norm.

The paper is organized as follows. Section 2 contains definitions
and formulation of the main results, remainder sections contain detailed
proofs.

\section{Model and results}

\begin{lemma} If $q(0),p(0)\in l_{2}(\mathbb{Z})$ then there exists
unique solution $q(t),p(t)$ of (\ref{mainEqForQ}) which belongs
to $l_{2}(\mathbb{Z})$, i.e.\ $q(t),p(t)\in l_{2}(\mathbb{Z})$ for
all $t\geqslant0$. \end{lemma}

\begin{proof} This assertion is well known (see \cite{LanfordLebowitz,DalKrein,Deimling}),
and easily follows from the boundedness of the operator $W$ on $l_{2}(\mathbb{Z})$:
\[
(Wq)_{k}=-\omega_{0}^{2}q_{k}+\omega_{1}^{2}(q_{k+1}-q_{k})-\omega_{1}^{2}(q_{k}-q_{k-1}).
\]
\end{proof}

\subsection{Uniform boundedness}

The first question of our interest is uniform boundedness (in $k$
and time $t\geqslant0$) of $|q_{k}(t)|$. Define the max-norm of
$q_{k}(t)$: 
\[
M(t)=\sup_{k}|q_{k}(t)|.
\]

We shall say that the system has the property of uniform boundedness
if: 
\[
\sup_{t\geqslant0}M(t)<\infty.
\]

\begin{theorem} \label{uniBoundTh} The following assertions hold:
\begin{enumerate} 
\item If $\omega_{0}>0$, then: 
\[
\sup_{t\geqslant0}M(t)<\infty.
\]
\item If $\omega_{0}=0$ then we have the following results.
  \begin{enumerate} 
\item For all $t\geqslant0$ the following inequality holds: 
\begin{equation}
M(t)\leqslant\frac{2}{\sqrt{\omega_{1}}}||p(0)||_{2}\sqrt{t}+||q(0)||_{2} . \label{MtInEq}
\end{equation}
\item Suppose that 
\begin{equation}
\sum_{k\ne0}|p_{k}(0)|\ln|k|<\infty.\label{omzeroPcond}
\end{equation}
Then there is a constant $c>0$ such that for all $t\geqslant1$: 
\[
M(t)\leqslant\frac{\sqrt{2}}{\omega_{1}\pi}|P|\ln(t)+||q(0)||_{2}+c,\quad P=\sum_{k}p_{k}(0).
\]
\item For all $\delta>1/2$ there exists at least one initial condition
$q(0)=0,p(0)\in l_{2}(\mathbb{Z})$ such that 
\[
\lim_{t\rightarrow\infty}\frac{q_{0}(t)}{\sqrt{t}}\ln^{\delta}t=\frac{\Gamma(\delta)}{\sqrt{2\omega_1}}>0
\]
where $\Gamma$ is the gamma function.
\end{enumerate}
\end{enumerate}
\end{theorem}

From the case 2.a we see that if $\omega_{0}=0$ and initial velocities
of the particles are all zero, then $|q_{k}(t)|$ are uniformly bounded.

The assertions 2.c is an attempt to answer the question on 
 the accuracy in the basic inequality (\ref{MtInEq}) from 2.a with
respect to the rate of growth in $t$.

\subsection{Asymptotic behavior}

Next we will formulate theorems concerning asymptotic behavior of
$q_{k}(t)$ in several cases.

Define Fourier transform of the sequence $u=\{u_{k}\}\in l_{2}(\mathbb{Z})$:
\[
\widehat{u}(\lambda)=\sum_{k}u_{k}e^{ik\lambda},\ \lambda\in\mathbb{R}.
\]
Note that $\widehat{u}(\cdot)\in L_{2}([0,2\pi])$, i.e.\ 
\[
\int_{0}^{2\pi}|\widehat{u}(\lambda)|^{2}\ d\lambda=2\pi\sum_{k}|u_{k}|^{2}<\infty.
\]

Further on we will use the Fourier transform of the initial conditions:
\[
Q(\lambda)=\widehat{q(0)}(\lambda),\quad P(\lambda)=\widehat{p(0)}(\lambda).
\]

For complex valued functions $f,g$ on $\mathbb{R}$ and constant
$c\in\mathbb{C}$ we will write $f(x)\asymp c+ g(x) / \sqrt{x}$,
if $f(x)=c+g(x) / \sqrt{x} +\bar{\bar{o}}(1 / \sqrt{x})$
as $x\rightarrow\infty$. 

\begin{theorem}[$\boldsymbol{\omega_0>0}$] \label{asympbehgz} Suppose that $\omega_{0}>0$
and $Q,P$ are of class $C^{n}(\mathbb{R})$ for some $n\geqslant2$.
Then
\begin{enumerate} 
\item For any fixed $t\geqslant0$ we have $q_{k}(t)=O(k^{-n})$. 
\item For any fixed $k\in\mathbb{Z}$ and $t\rightarrow\infty$ we
have the following asymptotic formula: 
\begin{align*}
  q_{k}(t)&\asymp\frac{1}{\sqrt{t}}\bigl(C_{1}\cos(\omega_{1}(t))+S_{1}\sin(\omega_{1}(t))\\
  &\quad {} +(-1)^{k}C_{2}\cos(\omega_{2}(t))+(-1)^{k}S_{2}\sin(\omega_{2}(t))),
\end{align*}
where 
\[
C_{1}=\frac{1}{\omega_{1}}\sqrt{\frac{\omega_{0}}{2\pi}}Q(0),\quad S_{1}=\frac{1}{\omega_{1}\omega_{0}}\sqrt{\frac{\omega_{0}}{2\pi}}P(0)
\]
\[
C_{2}=\frac{1}{\omega_{1}}\sqrt{\frac{\omega'_{0}}{2\pi}}Q(\pi),\quad S_{2}=\frac{1}{\omega_{1}\omega'_{0}}\sqrt{\frac{\omega'_{0}}{2\pi}}P(\pi),
\]
\[
\omega_{1}(t)=t\omega_{0}+\frac{\pi}{4},\quad\omega_{2}(t)=t\omega'_{0}-\frac{\pi}{4},\quad\omega'_{0}=\sqrt{\omega_{0}^{2}+4\omega_{1}^{2}}.
\]

\item Let $t=\beta|k|,\ \beta>0$ and $k\rightarrow\infty$. Put
\[
\gamma(\beta)=\beta^{2}\omega_{1}^{2}-1-\beta\omega_{0}.
\]
\begin{enumerate} 
\item If $\gamma(\beta)>0$ then 
\[
q_{k}(t)\asymp\frac{1}{\sqrt{|k|}}\Bigl(\mathcal{F}_{k}^{+}[Q]-i\mathcal{F}_{k}^{-}\Bigl[\frac{P(\lambda)}{\omega(\lambda)}\Bigr]\Bigr)
\]
where we introduce the following functionals for a complex valued function $g(\lambda)$ defined on the real
line: 
\begin{align*}
  \mathcal{F}_{k}^{\pm}[g]&=c_{+}(g(\mu_{+})e^{i\omega_{+}(k)}\pm g(-\mu_{+})e^{-i\omega_{+}(k)})\\
  &\quad {} +c_{-}(g(\mu_{-})e^{i\omega_{-}(k)}\pm g(-\mu_{-})e^{-i\omega_{-}(k)}),
\\
  \omega_{\pm}(k)&=k(\mu_{\pm}+\beta\omega(\mu_{\pm}))\pm\frac{\pi}{4}\mathrm{sign}(k),\\
  c_{\pm}&=\frac{1}{2}\sqrt{\frac{\beta\omega(\mu_{\pm})}{2\pi\Delta}},\\
  \mu_{\pm}&=-\arccos\frac{1}{\beta^{2}\omega_{1}^{2}}(1\pm\Delta),\\ 
  \Delta&=\sqrt{(\beta^{2}\omega_{1}^{2}-1)^{2}-\beta^{2}\omega_{0}^{2}},\\
  \omega(\lambda)&=\sqrt{\omega_{0}^{2}+2\omega_{1}^{2}(1-\cos\lambda)}.
\end{align*}
\item If $\gamma(\beta)=0$ and $n\geqslant3$ then $q_{k}(t)=O(k^{-3})$. 
\item if $\gamma(\beta)<0$ then $q_{k}(t)=O(k^{-n})$ for $n$ defined
  above.
\end{enumerate}
\end{enumerate}
\end{theorem}

Recall that a sufficient condition on $z\in l_{2}(\mathbb{Z})$ for
$\widehat{z}\in C^{n}(\mathbb{R})$ is 
\[
\sum_{k}|k|^{n}|z_{k}|<\infty.
\]
Thus if the following series converge for some $n\geqslant2$: 
\[
\sum_{k}|k|^{n}|q_{k}(0)|<\infty,\quad\mathrm{and}\quad\sum_{k}|k|^{n}|p_{k}(0)|<\infty,
\]
then Theorem \ref{asympbehgz} holds.

\begin{theorem}[$\boldsymbol{\omega_0=0}$] \label{asympbehezero}
  Suppose that
$\omega_{0}=0$ and $Q,P\in C^{n}(\mathbb{R}),\ n\geqslant6$ then
\begin{enumerate} 

\item For any fixed $t\geqslant0$ we have $q_{k}(t)=O(k^{-n})$. 

\item For any fixed $k\in\mathbb{Z}$ and $t\rightarrow\infty$ one
has: 
\[
q_{k}(t)\asymp\frac{P(0)}{2\omega_{1}}+\frac{(-1)^{k}}{\sqrt{t}}\Bigl(C\cos\Bigl(2\omega_{1}t-\frac{\pi}{4}\Bigr)+S\sin\Bigl(2\omega_{1}t-\frac{\pi}{4}\Bigr)\Bigr),
\]
where 
\[
C=\frac{1}{\sqrt{\pi\omega_{1}}}Q(\pi),\quad S=\frac{1}{2\omega_{1}\sqrt{\pi\omega_{1}}}P(\pi).
\]
\end{enumerate}
\end{theorem}

\subsection{Remarks}

If $\omega_{0}=0$ then it makes sense to consider the displacement
variables: 
\[
z_{k}(t)=x_{k+1}(t)-x_{k}(t)-a,\quad u_{k}(t)=\dot{z}_{k}=v_{k+1}(t)-v_{k}(t),\quad k\in\mathbb{Z}.
\]
Suppose that $\{z_{k}(0)\}_{k\in\mathbb{Z}}\in l_{2}(\mathbb{Z})$
and $\{u_{k}(0)\}_{k\in\mathbb{Z}}\in l_{2}(\mathbb{Z})$. It is easy
to see that $z_{k}(t)$ solves (\ref{mainEqForQ}) with $\omega_{0}=0$.
Consequently all formulated assertions for $q_{k}(t)$ in the case
$\omega_{0}=0$ hold for variables $z_{k}(t)$. It is interesting
to note that the quantity $P$ from item 2.b of theorem \ref{uniBoundTh}
in terms of variables $z,u$ equals 
\[
P=\sum_{k}u_{k}(0)=\lim_{n\rightarrow\infty}(v_{n}(0)-v_{-n}(0)).
\]
So if 
\[
\sum_{k\ne0}|u_{k}|\ln|k|<\infty
\]
and initial velocities of the ``right'' particles and ``left''
particles are equal, i.e.\ $\lim_{n\rightarrow\infty}(v_{n}(0)-v_{-n}(0))=0$,
then from the case 2.b of  theorem \ref{uniBoundTh} follows uniform
boundedness of displacements $z_{k}(t)$.

\section{Proofs}

Let us introduce the energy (hamiltonian): 
\[
H=H(q,p)=\sum_{k}\frac{p_{k}^{2}}{2}+\frac{\omega_{0}^{2}}{2}\sum_{k}q_{k}^{2}+\frac{\omega_{1}^{2}}{2}\sum_{k}(q_{k}-q_{k-1})^{2}.
\]
One can easily check that the energy is conserved under the dynamics
(\ref{mainEqForQ}). It means that $H(q(t),p(t))=H(q(0),p(0))$ for
all $t\geqslant0$ where $q(t),p(t)$ solves (\ref{mainEqForQ}).
If $\omega_{0}>0$ then from the energy conservation law and the inequality:
\[
\sup_{k}|q_{k}(t)|\leqslant\frac{\sqrt{2H(q(t),p(t))}}{\omega_{0}}
\]
 the uniform boundedness of $|q_{k}(t)|$ follows, i.e. 
\[
\sup_{t\geqslant0}\sup_{k\in\mathbb{Z}}|q_{k}(t)|<\infty
\]
and so item (1) of Theorem \ref{uniBoundTh} is proved.

Let us analyze the Fourier transform of the solution (\ref{mainEqForQ}):
\[
\widehat{q(t)}(\lambda)=\sum_{k}q_{k}(t)e^{ik\lambda}.
\]
The inverse transformation is given by the formula: 
\begin{equation}
q_{k}(t)=\frac{1}{2\pi}\int_{0}^{2\pi}\widehat{q(t)}(\lambda)e^{-ik\lambda}d\lambda,\ k\in\mathbb{Z}.\label{invFourier}
\end{equation}

\begin{lemma} \label{qQPformulaLemma} The solution of (\ref{mainEqForQ})
can be expressed as 
\begin{equation}
q_{k}(t)=Q_{k}(t)+P_{k}(t),\label{qQPformula}
\end{equation}
where 
\begin{align*}
  Q_{k}(t)&=\frac{1}{2\pi}\int_{0}^{2\pi}Q(\lambda)\cos(t\omega(\lambda))e^{-ik\lambda}d\lambda,\\
  P_{k}(t)&=\frac{1}{2\pi}\int_{0}^{2\pi}P(\lambda)\frac{\sin(t\omega(\lambda))}{\omega(\lambda)}e^{-ik\lambda}d\lambda,
\\
  Q(\lambda)&=\widehat{q(0)}(\lambda),\quad P(\lambda)=\widehat{p(0)}(\lambda),\\
  \omega(\lambda)&=\sqrt{\omega_{0}^{2}+2\omega_{1}^{2}(1-\cos(\lambda))}.
\end{align*}
\end{lemma}

\begin{proof} Using (\ref{mainEqForQ}) we obtain : 
\begin{align*}
  \frac{d^{2}}{dt^{2}}\widehat{q(t)}(\lambda)&=-\omega_{0}^{2}\widehat{q(t)}(\lambda)+\omega_{1}^{2}\sum_{k}q_{k+1}(t)e^{ik\lambda}+\omega_{1}^{2}\sum_{k}q_{k-1}(t)e^{ik\lambda}-2\omega_{1}^{2}\widehat{q(t)}(\lambda)\\
  &=-\omega^{2}(\lambda)\widehat{q(t)}(\lambda).
\end{align*}
Thus $\widehat{q(t)}(\lambda)$ for fixed $\lambda$ is coordinate
of the harmonic oscillator with the frequency $\omega(\lambda)$ and
so the solution of equation for $\widehat{q(t)}(\lambda)$ is 
\[
\widehat{q(t)}(\lambda)=\widehat{q(0)}(\lambda)\cos(t\omega(\lambda))+\widehat{p(0)}(\lambda)\frac{\sin(t\omega(\lambda))}{\omega(\lambda))}.
\]
From this equality and the formula of the inverse transformation (\ref{invFourier})
lemma follows.
\end{proof}

\subsection{Proof of Theorem \ref{uniBoundTh}}

The case $\omega_{0}>0$ we have considered above. Now suppose that
$\omega_{0}=0$. We will use the representation (\ref{qQPformula})
and some upper bounds for $Q_{k},P_{k}$.

At first we will prove the part 2.a of Theorem \ref{uniBoundTh}.
From the Cauchy\tire Bunyakovsky\tire Schwarz inequality we obtain the inequalities:
\begin{align*}
  |Q_{k}(t)|&\leqslant\sqrt{\frac{1}{2\pi}\int_{0}^{2\pi}|Q(\lambda)|^{2}d\lambda}\sqrt{\frac{1}{2\pi}\int_{0}^{2\pi}\cos^{2}(t\omega(\lambda))d\lambda}\\
  &\leqslant\sqrt{\frac{1}{2\pi}\int_{0}^{2\pi}|Q(\lambda)|^{2}d\lambda}=||q(0)||_{2} ,
\\
  |P_{k}(t)|&\leqslant\sqrt{\frac{1}{2\pi}\int_{0}^{2\pi}|P(\lambda)|^{2}d\lambda}\sqrt{\frac{1}{2\pi}\int_{0}^{2\pi}\frac{\sin^{2}(t\omega(\lambda))}{\omega^{2}(\lambda))}d\lambda}\\
  &=||p(0)||_{2}\sqrt{\frac{1}{2\pi}\int_{0}^{2\pi}\frac{\sin^{2}(t\omega(\lambda))}{\omega^{2}(\lambda))}d\lambda}.
\end{align*}
In the case $\omega_{0}=0$ one has $\omega(\lambda)=\sqrt{2\omega_{1}^{2}(1-\cos(\lambda))}=2\omega_{1}\sin(\lambda/2)$
and 
\begin{align*}
  I&=\int_{0}^{2\pi}\frac{\sin^{2}(t\omega(\lambda))}{\omega^{2}(\lambda)}d\lambda=2\int_{0}^{\pi}\frac{\sin^{2}(2\omega_{1}t\sin(u))}{(2\omega_{1}\sin(u))^{2}}du\\
  &=4\int_{0}^{\pi / 2}\frac{\sin^{2}(2\omega_{1}t\sin(u))}{(2\omega_{1}\sin(u))^{2}}du.
\end{align*}
Substituting $x=\sin u$ in the last integral we get: 
\begin{equation}
  I=\frac{1}{\omega_{1}^{2}}\int_{0}^{1}\frac{\sin^{2}(2\omega_{1}tx)}{x^{2}}\frac{1}{\sqrt{1-x^{2}}}dx
  =\frac{1}{\omega_{1}^{2}}\biggl(\int_{0}^{1 / \sqrt{2}}\ldots dx+\int_{1 / \sqrt{2}}^{1}\ldots dx\biggr). \label{Irepresent}
\end{equation}
The integrals in the latter formula will be estimated separately.
\begin{align*}
  \int_{0}^{1 / \sqrt{2}} \frac{\sin^{2}(2\omega_{1}tx)}{x^{2}}\frac{1}{\sqrt{1-x^{2}}}dx&\leqslant\sqrt{2}\int_{0}^{1 / \sqrt{2}}
                                                                                           \frac{\sin^{2}(2\omega_{1}tx)}{x^{2}}dx\\
  &=2\omega_{1}t\sqrt{2}\int_{0}^{\sqrt{2}\omega_{1}t}\frac{\sin^{2}(x)}{x^{2}}dx
\\
                                                                                         &\leqslant2\omega_{1}t\sqrt{2}\int_{0}^{\infty}\frac{\sin^{2}(x)}{x^{2}}dx\\
  &=2\omega_{1}t\sqrt{2}\frac{\pi}{2}=\pi\omega_{1}t\sqrt{2}.
\end{align*}
 One can find the value of the last integral in \cite{GR}, p.\thinspace 713,
3.821 (9). For the second integral in (\ref{Irepresent}) we have
\begin{align*}
  \int\limits_{1/ \sqrt{2}}^{1}\frac{\sin^{2}(2\omega_{1}tx)}{x^{2}}\frac{1}{\sqrt{1-x^{2}}}dx&\leqslant
                                                                                                2\omega_{1}t\int\limits_{1 / \sqrt{2}}^{1}\frac{|\sin(2\omega_{1}tx)|}{x}\frac{1}{\sqrt{1-x^{2}}}dx\\
  &\leqslant2\omega_{1}t\sqrt{2}\int\limits_{1/ \sqrt{2}}^{1}\frac{1}{\sqrt{1-x^{2}}}dx
=2\omega_{1}t\sqrt{2}\frac{\pi}{4}=\pi\omega_{1}t\frac{1}{\sqrt{2}}.
\end{align*}
Thus we obtain the inequality for $P_{k}$: 
\[
|P_{k}(t)|\leqslant\sqrt{\frac{1}{2\pi\omega_{1}^{2}}\left(\pi\omega_{1}t\sqrt{2}+\pi\omega_{1}t\frac{1}{\sqrt{2}}\right)}||p(0)||_{2}\leqslant2\sqrt{\frac{t}{\omega_{1}}}||p(0)||_{2}.
\]
This proves the case 2.a of Theorem \ref{uniBoundTh}.

Next we will check assertion 2.b. Condition (\ref{omzeroPcond})
implies $p(0)\in l_{1}(\mathbb{Z})$ and consequently $P(\lambda)=\widehat{p(0)}(\lambda)$
is a bounded continuous function on $\mathbb{R}$. We have the following
representation of $P_{k}(t)$: 
\begin{equation}
P_{k}(t)=\frac{1}{2\pi}\int_{0}^{2\pi}\frac{P(\lambda)-P(0)}{\omega(\lambda)}\sin(t\omega(\lambda))e^{-ik\lambda}d\lambda+\frac{P(0)}{2\pi}\int_{0}^{2\pi}\frac{\sin(t\omega(\lambda))}{\omega(\lambda)}e^{-ik\lambda}d\lambda.\label{PkRepresF}
\end{equation}
We estimate the first integral using condition (\ref{omzeroPcond}):
\begin{align*}
  I_{k}&=\biggl|\int_{0}^{2\pi}\frac{P(\lambda)-P(0)}{\omega(\lambda)}\sin(t\omega(\lambda))e^{-ik\lambda}d\lambda\biggr|\\
  &\leqslant\int_{0}^{2\pi}\Bigl| \frac{P(\lambda)-P(0)}{\omega(\lambda)}\Bigr| d\lambda
\leqslant\sum_{j}|p_{j}(0)|\int_{0}^{2\pi}\Bigl| \frac{e^{ij\lambda}-1}{2\omega_{1}\sin( \lambda / 2)}\Bigr| d\lambda .
\end{align*}
On the other hand 
\begin{align*}
  \int_{0}^{2\pi}\Bigl| \frac{e^{ij\lambda}-1}{\sin(\lambda / 2)}\Bigr| d\lambda&=
   \int_{0}^{2\pi}\Bigl| \frac{\sin(j\lambda / 2)}{\sin(\lambda / 2)}\Bigr| d\lambda=2\int_{0}^{\pi}\Bigl| \frac{\sin(ju)}{\sin u}\Bigr| d\lambda\\
  &=4\int_{0}^{\pi / 2}\Bigl| \frac{\sin(ju)}{\sin u}\Bigr| d\lambda\leqslant^{(1)}
    2\pi\int_{0}^{\pi / 2}\Bigl| \frac{\sin(ju)}{u}\Bigr| d\lambda\\
  &\leqslant^{(2)}2\pi(\ln|j|+c).
\end{align*}
In the inequality $^{(1)}$ we have used the fact that $\sin x\geqslant (2 / \pi) x$
for $x\in[0; \pi / 2]$ and $^{(2)}$ follows from Lemma \ref{absSinIneq}
below. Thus we have obtained 
\[
I_{k}\leqslant\frac{\pi}{\omega_{1}}\sum_{j\ne0}|p_{j}(0)|(\ln|j|+c)<\infty.
\]
Further we will estimate the second integral in (\ref{PkRepresF}):
\begin{align*}
  J_{k}&=\biggl| \int_{0}^{2\pi}\frac{\sin(t\omega(\lambda))}{\omega(\lambda)}e^{-ik\lambda}d\lambda\biggr| \leqslant
         2\int_{0}^{\pi}\frac{|\sin(2\omega_{1}t\sin u)|}{2\omega_{1}\sin u}du\\
  &=\frac{2}{\omega_{1}}\int_{0}^{1}\frac{|\sin(2\omega_{1}tx)|}{x}\frac{1}{\sqrt{1-x^{2}}}dx
\\
&\leqslant\frac{2}{\omega_{1}}\biggl(\int_{0}^{1 / \sqrt{2}}\ldots dx+\int_{1 / \sqrt{2}}^{1}\ldots dx\biggr) \leqslant\frac{2}{\omega_{1}}(\sqrt{2}\ln t+c)
\end{align*}
for some constant $c>0$ not depending on $t\geqslant 1$ and $k$. In the latter
inequality we again have applied lemma \ref{absSinIneq}.

Finally from (\ref{PkRepresF}) and due to the bounds for $I_{k},J_{k}$
we obtain: 
\[
|P_{k}(t)|\leqslant\frac{\sqrt{2}}{\omega_{1}\pi}|P(0)|\ln t+c
\]
for some constant $c>0$ not depending on $t\geqslant 1$ and $k$. This completes
the proof of the part 2.b.

\begin{lemma} \label{absSinIneq} For all $b>0$ and  all $t_0>0$ there is a constant
$c>0$ such that for any $t\geqslant t_0$ the following inequality holds:
\[
\int_{0}^{b}\frac{|\sin(tx)|}{x}dx\leqslant\ln t+c.
\]
\end{lemma}
\begin{proof} Substituting $tx=y$ we get: 
\begin{align*}
  \int_{0}^{b}\frac{|\sin(tx)|}{x}dx&=\int_{0}^{tb}\frac{|\sin y|}{y}dy=\int_{0}^{t_0 b}\frac{|\sin y|}{y}dy+\int_{t_0 b}^{tb}\frac{|\sin y|}{y}dy\\
  &\leqslant c+\int_{t_0 b}^{tb}\frac{1}{y}dy=\ln -\ln t_0+c.
\end{align*}
This proves the assertion. \end{proof}

Finally we will prove part 2.c of Theorem \ref{uniBoundTh}. We
will construct the required initial condition in two steps. At first
step for any $0<\alpha<1/2$ we will find initial conditions (depending
on $\alpha$) such that the corresponding solution satisfies $\lim_{t\rightarrow\infty}q_{0}^{(\alpha)}(t)/t^{\alpha}>0$.
At the next step we will integrate these initial conditions with an
appropriate weight and prove that the resulting function gives us
the answer.

Firstly we prove the assertion if $\omega_{1}=1/2$.
Consider initial conditions $q^{(\alpha)}(0),p^{(\alpha)}(0)$ with
the following Fourier transforms: 
\[
Q^{(\alpha)}(\lambda)=0,\quad P^{(\alpha)}(\lambda)=\frac{a_{\alpha}}{(|\omega(\lambda)|)^{\alpha}}=\frac{a_{\alpha}}{|\sin \lambda / 2|^{\alpha}}
\]
where $0<\alpha< 1/2$ and the constant $a_{\alpha}>0$ is
chosen so that 
\[
||P^{(\alpha)}(\lambda)||_{L_{2}([0,2\pi])}^{2}=\int_{0}^{2\pi}|P^{(\alpha)}(\lambda)|^{2}d\lambda=1.
\]
 Exact formula for $a_{\alpha}$  will be given  below. It is obvious
that $P^{(\alpha)}(\lambda) \! \in \! L_{2}([0,2\pi])$. So the corresponding
initial conditions 
\[
q_{k}^{(\alpha)}(0)=0,\quad p_{k}^{(\alpha)}(0)=\frac{1}{2\pi}\int_{0}^{2\pi}P^{(\alpha)}(\lambda)e^{-ik\lambda}d\lambda
\]
lie in $l_{2}(\mathbb{Z})$. From (\ref{qQPformula})
we have 
\[
q_{0}(t)=q_{0}^{(\alpha)}(t)=P_{0}(t)=\frac{a_{\alpha}}{2\pi}\int_{0}^{2\pi}\frac{\sin(t\sin \lambda / 2)}{|\sin \lambda / 2|^{\alpha+1}}d\lambda.
\]

\begin{lemma} \label{qalphatlemma} For all $t>0$ the following
equality holds: 
\[
q_{0}^{(\alpha)}(t)=\varphi(\alpha)t^{\alpha}+R(\alpha,t),\quad\varphi(\alpha)=2a_{\alpha}\frac{\Gamma(1-\alpha)}{\pi\alpha}\cos\frac{\pi\alpha}{2},
\]
where for the remainder term $R$ we have 
\[
|R(\alpha,t)|\leqslant a_{\alpha}\Bigl(3+\frac{2}{t}\Bigr),
\]
where $\Gamma$ is the gamma function.
\end{lemma}
\begin{proof}
From the definition we have 
\[
  q_{0}^{(\alpha)}(t)=\frac{a_{\alpha}}{\pi}\int_{0}^{\pi}\frac{\sin(t\sin\lambda)}{\sin^{\alpha+1}\lambda}d\lambda=
  \frac{a_{\alpha}}{\pi} \Bigl( I\Bigl[ 0,\frac{\pi}{4}\Bigr]+I\Bigl[\frac{\pi}{4},\frac{3\pi}{4}\Bigr]+I\Bigl[\frac{3\pi}{4},\pi\Bigr]\Bigr),
\]
where 
\[
I[a,b]=\int_{a}^{b}\frac{\sin(t\sin\lambda)}{\sin^{\alpha+1}\lambda}d\lambda.
\]
\par 
The second integral can easily be estimated: 
\[
  \Big| I \Bigl[\frac{\pi}{4},\frac{3\pi}{4}\Bigr] \Big| \leqslant\int_{\pi / 4}^{3\pi/4}
  \frac{d\lambda}{(\sin\lambda)^{\alpha+1}}\leqslant\frac{\pi}{2}(\sqrt{2})^{\alpha+1}\leqslant\pi2^{-1/4}<4.
\]
The third integral $I[ 3\pi / 4,\pi]$ evidently equals to the
first one $I[0, \pi / 4]$ (to see this it is sufficient to make
the substitution $x=\pi-\lambda$). Let us evaluate the first integral.
Substituting $x=\sin\lambda$ we have: 
\begin{align}
  I\Bigl[0,\frac{\pi}{4}\Bigr]&=\int_{0}^{1/ \sqrt{2}} \frac{\sin(tx)}{x^{\alpha+1}\sqrt{1-x^{2}}}dx \label{Izeropifour}   \\
                    &=\int_{0}^{1 / \sqrt{2}} \frac{\sin(tx)}{x^{\alpha+1}}dx+\int_{0}^{1 / \sqrt{2}} \Bigl(\frac{\sin(tx)}{x^{\alpha+1}\sqrt{1-x^{2}}}-\frac{\sin(tx)}{x^{\alpha+1}}\Bigr) dx. \nonumber
                      \end{align}
Due to the mean-value theorem for all $0\leqslant s\leqslant 1/2$
the following inequality holds: 
\[
\Bigl| \frac{1}{\sqrt{1-s}}-1\Bigr| \leqslant s\max_{0\leq\theta\leq s}\frac{1}{2(\sqrt{1-\theta})^{3}}\leqslant s\frac{  1} {2 \bigl(\sqrt{1- 1/2}\bigr)^{3}}=s\sqrt{2}.
\] 
Consequently for the second integral in (\ref{Izeropifour}) we have
the estimates: 
\begin{align*}
  \biggl| \int_{0}^{1 / \sqrt{2}} \Bigl(\frac{\sin(tx)}{x^{\alpha+1}\sqrt{1-x^{2}}}-\frac{\sin(tx)}{x^{\alpha+1}}\Bigr) dx\biggr| &
       \leqslant\int_{0}^{1 / \sqrt{2}} \frac{|\sin(tx)|}{x^{\alpha+1}}\Bigl| \frac{1}{\sqrt{1-x^{2}}}-1\Bigr| dx \\
  &\leqslant\int_{0}^{1 / \sqrt{2}} \frac{\sqrt{2}x^{2}}{x^{\alpha+1}}dx
\\
&=\sqrt{2}\frac{1}{2-\alpha}\frac{1}{(\sqrt{2})^{2-\alpha}}\leqslant\frac{1}{(\sqrt{2})^{1-\alpha}}<2.
\end{align*}
The first integral in (\ref{Izeropifour}) can be expressed as: 
\[
  \int_{0}^{1 / \sqrt{2}}\frac{\sin(tx)}{x^{\alpha+1}}dx=t^{\alpha}\int_{0}^{t / \sqrt{2}}
  \frac{\sin u}{u^{\alpha+1}}du=t^{\alpha}\int_{0}^{+\infty}\frac{\sin u}{u^{\alpha+1}}du-t^{\alpha}\int_{t / \sqrt{2}}^{+\infty}\frac{\sin u}{u^{\alpha+1}}du.
\]
In the latter formula the first integral is Bohmer integral (generalized
Fresnel integral) which value can be found in \cite{GR}, p.\thinspace 648,
3.712: 
\begin{align*}
  \int_{0}^{\infty}\frac{\sin u}{u^{\alpha+1}}du&=\frac{1}{\alpha(1-\alpha)}\int_{0}^{\infty}\cos y^{1/(1-\alpha)}\,dy\\
  &=\frac{1}{\alpha(1-\alpha)}\frac{\Gamma(1-\alpha)\sin(\frac{\pi}{2}(1-\alpha))}{1/(1-\alpha)}=\frac{\Gamma(1-\alpha)}{\alpha}\cos\frac{\pi\alpha}{2}.
\end{align*}
Integrating by parts we estimate the remainder term: 
\begin{align*}
  \biggl| \int_{t / \sqrt{2}}^{+\infty} \frac{\sin u}{u^{\alpha+1}}du\biggr| &=
                                                                               \biggl| \frac{\cos(t / \sqrt{2})}{(t / \sqrt{2})^{\alpha+1}}
                                                                               -(1+\alpha)\int_{t/ \sqrt{2}}^{\infty}\frac{\cos u}{u^{\alpha+2}}du\biggr| \\
                &\leqslant\frac{2}{t^{\alpha+1}}+(1+\alpha)\int_{t / \sqrt{2}}^{\infty}\frac{1}{u^{\alpha+2}}du\leqslant\frac{4}{t^{\alpha+1}}.
\end{align*}
These inequalities complete the proof.
\end{proof}

For any $0<\varepsilon<1/2$ define a weight function 
\[
w_{\varepsilon}(\alpha)=\frac{1}{\varphi(\alpha)}\ \frac{1}{( 1/2-\alpha)^{1/2-\varepsilon}}
\]
where $\varphi(\alpha)$ is defined in Lemma \ref{qalphatlemma}.
\begin{lemma} $w_{\varepsilon}(\alpha)$ is absolutely integrable
w.r.t.\ $\alpha$ on $[0, 1/2]$: 
\[
\int_{0}^{1/2}w_{\varepsilon}(\alpha)d\alpha<\infty.
\]
\end{lemma}
\begin{proof}
  Now we need the exact expression of $a_{\alpha}$:
\begin{align*}
  \frac{1}{a_{\alpha}^{2}}&=\int_{0}^{2\pi}\frac{1}{(\sin\frac{\lambda}{2})^{2\alpha}}d\lambda=4\int_{0}^{\pi / 2}
                            \frac{1}{\sin^{2\alpha}x}dx\\
  &=2\mathrm{B}\Bigl(\frac{1-2\alpha}{2},\frac{1}{2}\Bigr)=2\frac{\Gamma(\frac{1}{2}-\alpha)\Gamma(\frac{1}{2})}{\Gamma(1-\alpha)}
\end{align*}
where $B$ is the beta function (\hspace{-0.1pt}\cite{GR}, p.\thinspace 610, 3.621). Hence
\[
a_{\alpha}=\sqrt{\frac{\Gamma(1-\alpha)}{2\sqrt{\pi}\Gamma(\frac{1}{2}-\alpha)}}.
\]
It is well-known that $\Gamma(z)= 1/z+O(1)$ as $z\rightarrow0$.
Thus $w_{\varepsilon}(\alpha)$ has the only one singular point on
$[0, 1/2]$ at $\alpha=1/2$ and obviously 
\[
w_{\varepsilon}(\alpha)\sim\frac{c}{(1/2-\alpha)^{1-\epsilon}}
\]
as $\alpha\rightarrow  1/2$ for some constant $c$. So $w_{\varepsilon}(\alpha)$
is absolutely integrable on $[0, 1/2]$.
\end{proof}

Finally, we will construct required initial condition by its Fourier
transform which are 
 defined by the following formulas: 
\[
\tilde{Q}^{(\varepsilon)}(\lambda)=0,\quad\tilde{P}^{(\varepsilon)}(\lambda)=\int_{0}^{1/2}w_{\varepsilon}(\alpha)P^{(\alpha)}(\lambda)d\alpha.
\]
The latter integral we understand in the following sense: 
\begin{equation}
\tilde{P}^{(\varepsilon)}(\lambda)=\lim_{\delta\rightarrow0+}\int_{0}^{1/2-\delta}w_{\varepsilon}(\alpha)P^{(\alpha)}(\lambda)d\alpha.\label{limdDef}
\end{equation}
Since $L_{2}$-norm of $P^{(\alpha)}(\cdot)$ equals to one and $w_{\varepsilon}(\alpha)$
is absolutely integrable, the limit in (\ref{limdDef}) exists and,
moreover, one has the inequality: 
\[
||\tilde{P}^{(\varepsilon)}(\cdot)||_{L_{2}([0,2\pi])}\leqslant\int_{0}^{1/2}w_{\varepsilon}(\alpha)d\alpha.
\]
Thus the corresponding to $\tilde{Q}^{(\varepsilon)}(\lambda),\tilde{P}^{(\varepsilon)}(\lambda)$
initial conditions: 
\begin{equation}
\tilde{q}_{k}^{(\varepsilon)}(0)=0,\quad\tilde{p}_{k}^{(\varepsilon)}(0)=\frac{1}{2\pi}\int_{0}^{2\pi}\tilde{P}^{(\varepsilon)}(\lambda)e^{-ik\lambda}d\lambda\label{epsInitCond}
\end{equation}
lie in $l_{2}(\mathbb{Z})$. Denote $\tilde{q}_{k}^{(\varepsilon)}(t),\tilde{p}_{k}^{(\varepsilon)}(t),\ k\in\mathbb{Z}$
the solution of (\ref{mainEqForQ}) with initial condition (\ref{epsInitCond}).
Due to the lemmas \ref{qQPformulaLemma} and \ref{qalphatlemma} and
Fubini\tire Tonelli theorem we have: 
\begin{align*}
  \tilde{q}_{0}^{(\varepsilon)}(t)&=\frac{1}{2\pi}\int_{0}^{2\pi}\frac{\sin(t\sin\frac{\lambda}{2})}{\sin\frac{\lambda}{2}}\tilde{P}^{(\varepsilon)}(\lambda)d\lambda\\
  &=\int_{0}^{1/2}w_{\varepsilon}(\alpha)\frac{1}{2\pi}\int_{0}^{2\pi}\frac{\sin(t\sin\frac{\lambda}{2})}{\sin\frac{\lambda}{2}}P^{(\alpha)}(\lambda)d\lambda d\alpha
\\
                                  &=\int_{0}^{1/2} \! w_{\varepsilon}(\alpha)q_{0}^{(\alpha)}(t)d\alpha= \! \int_{0}^{1/2} \! \! \frac{1}{\varphi(\alpha)}\,
                                    \frac{1}{(1/2-\alpha)^{1/2-\varepsilon}}\bigl(\varphi(\alpha)t^{\alpha} \! + \! R(\alpha,t)\bigr)d\alpha
\\
                                  &=\int_{0}^{1/2}\frac{t^{\alpha}}{(1/2-\alpha)^{1/2-\varepsilon}}d\alpha+\int_{0}^{1/2}\frac{R(\alpha,t)}{\varphi(\alpha)}\
                                    \frac{1}{(1/2-\alpha)^{1/2-\varepsilon}}d\alpha .
\end{align*}
The remainder term in the latter formula can be easily estimated:
\[
\biggl| \int_{0}^{1/2} \frac{R(\alpha,t)}{\varphi(\alpha)}\ \frac{1}{( 1/2 -\alpha)^{1/2-\varepsilon}}d\alpha\biggr| \leqslant c_{1}+\frac{c_{2}}{t}
\]
for some nonnegative constants $c_{1},c_{2}$. Let us find the value
of the first integral. Put $\delta=\varepsilon+1/2$: 
\begin{align*}
  \int_{0}^{1/2} \frac{t^{\alpha}}{(1/2-\alpha)^{1-\delta}}d\alpha&=\sqrt{t}\int_{0}^{1/2} t^{-u}u^{\delta-1}du=\sqrt{t}\int_{0}^{1/2}e^{-u\ln t}u^{\delta-1}du\\
  &=\sqrt{t}(\ln t)^{-\delta}\int_{0}^{(\ln t)/2} e^{-y}y^{\delta-1}dy.
\end{align*}
Thus we have proved Theorem \ref{uniBoundTh} item 2.c for the case $\omega_1 =\frac{1}{2}$. 

 Now suppose that $\omega_1$ is an arbitrary positive number. 
Consider solution with initial condition $\tilde{q}_{k}^{(\varepsilon)}(0), \tilde{p}_{k}^{(\varepsilon)}(0)$ defined in (\ref{epsInitCond}). 
Denote it by $\tilde{q}_{k}^{\omega_1,(\varepsilon)}(t)$. From (\ref{qQPformula}) it is easy to see that 
$$
\tilde{q}_{k}^{\omega_1,(\varepsilon)}(t) = \frac{1}{2\omega_1}\tilde{q}_{k}^{1/2,(\varepsilon)}(2\omega_1 t) = \frac{1}{2\omega_1}\tilde{q}_{k}^{(\varepsilon)}(2\omega_1 t).
$$
Hence we obtain limiting equalities:
$$
\lim_{t\rightarrow\infty}\frac{\tilde{q}_{k}^{\omega_1,(\varepsilon)}(t)}{\sqrt{t}}\ln^{\delta}t = 
\frac{1}{2\omega_1} \lim_{t\rightarrow\infty}\frac{\tilde{q}_{k}^{(\varepsilon)}(2\omega_1 t)}{\sqrt{t}}\ln^{\delta}t  = \frac{1}{\sqrt{2\omega_1}} \Gamma(\delta).
$$
This completes the proof of Theorem \ref{uniBoundTh}.

\subsection{Proof of Theorem \ref{asympbehgz}}

We will use Lemma \ref{qQPformulaLemma}. The first part  of Theorem \ref{asympbehgz}
can be easily derived by integrating by parts $n$ times the following
integral: 
\[
\int_{0}^{2\pi}f(\lambda)e^{-ik\lambda}d\lambda
\]
where $f(\lambda)$ is a corresponding $C^{n}$ smooth $2\pi$-periodic
function.

Next we will need the following lemma. \begin{lemma} \label{egaymp}
Consider the integral: 
\[
E[g](t)=\frac{1}{2\pi}\int_{0}^{2\pi}g(\lambda)\exp(it\omega(\lambda))d\lambda,\quad \omega(\lambda)=\sqrt{\omega_{0}^{2}+2\omega_{1}^{2}(1-\cos\lambda)}
\]
where $g(\lambda)\in C^{n}(\mathbb{R}),\ n\geqslant2$ is a complex
valued $2\pi$-periodic function. Then,  as $t \rightarrow\infty,$
\[
  E[g](t)\asymp\frac{1}{\sqrt{t}}\left(c_{1}g(0)\exp\Bigl(i\Bigl(t\omega_{0}+\frac{\pi}{4}\Bigr)\Bigr)+c_{2}g(\pi)\exp\Bigl(i\Bigl(t\omega'_{0}-\frac{\pi}{4}\Bigr)
  \Bigr) \right),
\]
where 
\[
c_{1}=\frac{1}{\omega_{1}}\sqrt{\frac{\omega_{0}}{2\pi}},\quad c_{2}=\frac{1}{\omega_{1}}\sqrt{\frac{\omega'_{0}}{2\pi}}
\]
and $\omega_{0}'$ is defined in Theorem \ref{asympbehgz}.
\end{lemma}
\begin{proof} We will apply stationary phase method (see \cite{Erdelyi,Fedoruk}).
Let us find critical points of the phase function, i.e.\ zeros of the
$\frac{d}{d\lambda}\omega(\lambda)$: 
\[
\omega'(\lambda)=\frac{d}{d\lambda}\omega(\lambda)=\frac{\omega_{1}^{2}\sin\lambda}{\omega(\lambda)}=0.
\]
So the critical points lying on the interval $[0,2\pi]$ are only
$0,\pi,2\pi$. Zero and $2\pi$ are the boundary points, but as $g(\lambda)$
is a $2\pi$-periodic function we can replace the interval $[0,2\pi]$
from the definition of $E[g](t)$ by the interval $[- \pi / 2, 3\pi / 2]$: 
\[
E[g](t)=\frac{1}{2\pi}\int_{-\pi / 2}^{3\pi / 2}g(\lambda)\exp(it\omega(\lambda))d\lambda.
\]
Thus the critical points of $\omega(\lambda)$ lying on $[-\pi / 2, 3\pi / 2]$
are only $0,\pi$ and we can apply stationary phase method for the
internal critical points. Let us find the second derivative of the
$\omega(\lambda)$: 
\[
\omega''(\lambda)=\frac{\omega_{1}^{2}\cos\lambda}{\omega(\lambda)}-\frac{\omega_{1}^{4}\sin^{2}\lambda}{\omega^{3}(\lambda)}.
\]
Consequently $\omega''(0)=\omega_{1}^{2} / \omega_{0},$ $\omega''(\pi)=\omega_{1}^{2} / \omega'_{0}$.
To prove the lemma it remains to apply the formula (2) from \cite{Erdelyi},
p.\thinspace 51 (or \cite{Fedoruk}, p.\thinspace 163).
\end{proof}

Further we will use the equality (\ref{qQPformula}). For $Q_{k}$
we have: 
\begin{align*}
Q_{k}(t)&=\frac{1}{2\pi}\int_{0}^{2\pi}Q(\lambda)\cos(t\omega(\lambda))e^{-ik\lambda}d\lambda=\frac{1}{2}(E[Qe^{-ik\lambda}](t)+\overline{E[\overline{Qe^{-ik\lambda}}](t)}),
\\
  P_{k}(t)&=\frac{1}{2\pi}\int_{0}^{2\pi}P(\lambda)\frac{\sin(t\omega(\lambda))}{\omega(\lambda)}e^{-ik\lambda}d\lambda=\frac{1}{2i}\left(E[g](t)-\overline{E[\overline{g}](t)}\right),
\\
g&=\frac{Pe^{-ik\lambda}}{\omega(\lambda)}
\end{align*}
where $\overline{z}$ denotes the complex conjugate number of $z\in\mathbb{C}$.
Applying Lemma \ref{egaymp} to these expressions one can easily obtain
the part 2 of Theorem \ref{asympbehgz}.

Let us prove part 3 of Theorem \ref{asympbehgz}. We need the
following lemma.
\begin{lemma} \label{FgbetakL}
  Consider the integral
\[
F[g](\beta,k)=\frac{1}{2\pi}\int_{0}^{2\pi}g(\lambda)e^{ik(\lambda+\beta\omega(\lambda))}d\lambda
\]
where $g$ satisfies the conditions of Lemma \ref{egaymp}, $\beta>0,\ k\in\mathbb{Z}$.
Define the constant: 
\[
\gamma(\beta)=\beta^{2}\omega_{1}^{2}-1-\beta\omega_{0}.
\]
The following assertions hold: \begin{enumerate} 

\item if $\gamma(\beta)>0$ then as $k\rightarrow\infty$: 
\begin{equation}
F[g](\beta,k)\asymp\frac{1}{\sqrt{|k|}}(c_{+}g(\mu_{+})e^{i\omega_{+}(k)}+c_{-}g(\mu_{-})e^{i\omega_{-}(k)})\label{fbetakasymp}
\end{equation}
where 
\begin{align*}
  \omega_{\pm}(k)&=k(\mu_{\pm}+\beta\omega(\mu_{\pm}))\pm\frac{\pi}{4}\mathrm{sign}(k),\\
  c_{\pm}&=\sqrt{\frac{\beta\omega(\mu_{\pm})}{2\pi\Delta}},\quad\mu_{\pm}=-\arccos\frac{1}{\beta^{2}\omega_{1}^{2}}(1\pm\Delta),
\\
\Delta&=\sqrt{(\beta^{2}\omega_{1}^{2}-1)^{2}-\beta^{2}\omega_{0}^{2}};
\end{align*}

\item if $\gamma(\beta)=0$ and $n\geqslant3$ then $F[g](\beta,k)=O(k^{-3})$; 

\item if $\gamma(\beta)<0$ then $F[g](\beta,k)=O(k^{-n})$.
\end{enumerate}
\end{lemma}
\begin{proof} We will again use the stationary phase
method. Consider the phase function: 
\[
h(\lambda)=\lambda+\beta\omega(\lambda).
\]
 Let us find the critical points: 
\begin{equation}
h'(\lambda)=1+\beta\omega_{1}^{2}\frac{\sin\lambda}{\sqrt{\omega_{0}^{2}+2\omega_{1}^{2}(1-\cos\lambda)}}=0.\label{hdereq}
\end{equation}
From this equation we have: 
\[
\omega_{0}^{2}+2\omega_{1}^{2}(1-\cos\lambda)=\beta^{2}\omega_{1}^{4}(1-\cos^{2}\lambda).
\]
Making the substitution $x=\cos\lambda$ rewrite the last equation:
\begin{equation}
\beta^{2}\omega_{1}^{4}x^{2}-2\omega_{1}^{2}x+(\omega_{0}^{2}+2\omega_{1}^{2}-\beta^{2}\omega_{1}^{4})=0.\label{criticalPointsEqX}
\end{equation}
The discriminant of this equation is: 
\begin{align*}
  D&=4\omega_{1}^{4}(1-\beta^{2}(\omega_{0}^{2}+2\omega_{1}^{2}-\beta^{2}\omega_{1}^{4}))
     =4\omega_{1}^{4}\bigl((\beta^{2}\omega_{1}^{2}-1)^{2}-\beta^{2}\omega_{0}^{2}\bigr)\\
  &=4\omega_{1}^{4}(\beta^{2}\omega_{1}^{2}-1-\beta\omega_{0})(\beta^{2}\omega_{1}^{2}-1+\beta\omega_{0}).
\end{align*}
\par Suppose that $\gamma(\beta)\geqslant0$. In that case the roots
of the equation (\ref{criticalPointsEqX}) are: 
\[
  x_{\pm}=\frac{1}{\beta^{2}\omega_{1}^{2}}\Bigl(1\pm\sqrt{(\beta^{2}\omega_{1}^{2}-1)^{2}-\beta^{2}\omega_{0}^{2}}\,\Bigr)
  =\frac{1}{\beta^{2}\omega_{1}^{2}}(1\pm\Delta).
\]
From the condition $\gamma(\beta)\geqslant0$ it follows that $\beta^{2}\omega_{1}^{2}-1>0$
and thus: 
\[
|x_{\pm}|<\frac{1}{\beta^{2}\omega_{1}^{2}}(1+|\beta^{2}\omega_{1}^{2}-1|)=1.
\]
Consequently the phase function has only two critical points $\lambda_{\pm}$
on $[0,2\pi]$: $\cos(\lambda_{\pm})=x_{\pm}$ with the additional
condition $\sin(\lambda_{\pm})<0$ which follows from (\ref{hdereq}):
\[
\lambda_{\pm}=2\pi-\arccos\frac{1}{\beta^{2}\omega_{1}^{2}}(1\pm\Delta)=2\pi+\mu_{\pm}.
\]
Obviously $\lambda_{+}\ne\lambda_{-}$ iff  $\gamma(\beta)>0$.
Moreover, $\lambda_{\pm}$ are internal points, i.e.\ $\lambda_{\pm}\in(0,2\pi)$.\par To
apply the stationary phase method we should find the signs of $h''(\lambda_{\pm})$.
Rewrite the derivative of the phase function: 
\[
h'(\lambda)=1+\beta\omega_{1}^{2}\frac{\sin\lambda}{\omega(\lambda)}.
\]
So at $\lambda_{\pm}$ we have: 
\[
\omega(\lambda_{\pm})=-\beta\omega_{1}^{2}\sin\lambda_{\pm}.
\]
Further we obtain 
\[
h''(\lambda)=\beta\omega_{1}^{2}\frac{\cos\lambda}{\omega(\lambda)}-\beta\omega_{1}^{4}\frac{\sin^{2}\lambda}{\omega^{3}(\lambda)}.
\]
Thus 
\begin{align*}
h''(\lambda_{\pm})&=-\frac{\cos\lambda_{\pm}}{\sin\lambda_{\pm}}+\frac{1}{\beta\omega_{1}^{2}}\frac{1}{\sin\lambda_{\pm}}=-\frac{1}{\sin\lambda_{\pm}}\left(\cos\lambda_{\pm}-\frac{1}{\beta^{2}\omega_{1}^{2}}\right)
\\
&=\pm\Bigl(-\frac{1}{\sin\lambda_{\pm}}\Bigr)\frac{1}{\beta^{2}\omega_{1}^{2}}\sqrt{(\beta^{2}\omega_{1}^{2}-1)^{2}-\beta^{2}\omega_{0}^{2}}=\pm\frac{1}{\beta\omega(\lambda_{\pm})}\Delta.
\end{align*}
Since $\omega(\lambda)>0$ we have $\mathrm{sign}(h''(\lambda_{\pm}))=\pm1$
under the condition $\gamma(\beta)>0$. Hence applying (2) from \cite{Erdelyi},
p.\thinspace 51 we get: 
\[
F[g](\beta,k)\asymp\frac{1}{\sqrt{|k|}}\bigl(c_{+}g(\lambda_{+})e^{i\tilde{\omega}_{+}(k)}+c_{-}g(\lambda_{-})e^{i\tilde{\omega}_{-}(k)}\bigr)
\]
where $c_{\pm}$ are defined in (\ref{fbetakasymp}) (we have used
the fact that $\omega(\lambda)=\omega(-\lambda)$). Note that: 
\begin{align*}
\tilde{\omega}_{+}(k)&=k(\lambda_{\pm}+\beta\omega(\lambda_{\pm}))\pm\frac{\pi}{4}\mathrm{sign}(k)=2\pi k+\omega_{\pm}(k),
\\
g(\lambda_{\pm})&=g(\mu_{\pm})
\end{align*}
where $\omega_{\pm}(k)$ is defined in (\ref{fbetakasymp}). Since
$k\in\mathbb{Z}$ we have $e^{i\tilde{\omega}_{\pm}(k)}=e^{i\omega_{\pm}(k)}$.
This completes the prove in the case $\gamma(\beta)>0$.\par If $\gamma(\beta)=0$
then $\lambda_{+}=\lambda_{-}$ and $h''(\lambda_{+})=0$, i.e.\ $\lambda_{+}$
is a degenerate critical point and so $F[g](\beta,k)=O(k^{-3})$ due
to \cite{Erdelyi}, p.\thinspace 52 (or \cite{Fedoruk}, p.\thinspace 163).

\par
Consider
the case $\gamma(\beta)<0$. We will prove that $h(\lambda)$ has
no critical points on $[0,2\pi]$. Rewrite the discriminant using
$\gamma$: 
\[
D=4\omega_{1}^{4}\gamma(\gamma+2\beta\omega_{0}).
\]
If $\gamma<0$ and $\gamma+2\beta\omega_{0}>0$ then $D<0$ and the
equation (\ref{criticalPointsEqX}) has no real roots. Suppose that
$\gamma+2\beta\omega_{0}\leqslant0$. In that case 
\[
\Delta^{2}=\frac{D}{4\omega_{1}^{4}}=((\gamma+2\beta\omega_{0})^{2}-2\beta\omega_{0}(\gamma+2\beta\omega_{0}))\geqslant(\gamma+2\beta\omega_{0})^{2}.
\]
For the roots $x_{\pm}$ we have the following estimates: 
\begin{align*}
  |x_{\pm}| &\geqslant\frac{1}{\beta^{2}\omega_{1}^{2}}|1-\Delta|\geqslant\frac{1}{\beta^{2}\omega_{1}^{2}}|1-|\gamma+2\beta\omega_{0}||=\frac{1}{\beta^{2}\omega_{1}^{2}}|1+\gamma+2\beta\omega_{0}|\\
  &=\frac{1}{\beta^{2}\omega_{1}^{2}}(\beta^{2}\omega_{1}^{2}+\beta\omega_{0})>1.
\end{align*}
Thus we have proved that $h(\lambda)$ has no critical points on $[0,2\pi]$.
Consequently, $F[g](\beta,k)=O(k^{-n})$. This completes the proof.
\end{proof}

Let us prove the remainder part of Theorem \ref{asympbehgz}. From
Lemma \ref{qQPformulaLemma} we have 
\begin{align*}
  Q_{k}(t)&=\frac{1}{2\pi}\int_{0}^{2\pi}Q(\lambda)\cos(t\omega(\lambda))e^{-ik\lambda}d\lambda\\
  &=\frac{1}{4\pi}\int_{0}^{2\pi}Q(\lambda)e^{it\omega(\lambda)-ik\lambda}d\lambda+\frac{1}{4\pi}\int_{0}^{2\pi}Q(\lambda)e^{-it\omega(\lambda)-ik\lambda}d\lambda.
\end{align*}
Put $t=\beta k$ and rewrite $Q_{k}(t)$: 
\begin{align*}
  Q_{k}(t)&=\frac{1}{4\pi}\int_{0}^{2\pi}Q(\lambda)e^{it\omega(\lambda)-ik\lambda}d\lambda+\frac{1}{2}F[Q](\beta,-k)\\
  &=\frac{1}{2}(F[Q^{*}](\beta,k)+F[Q](\beta,-k)),
\end{align*}
where we have used the following notation: 
\[
f^{*}(\lambda)=f(2\pi-\lambda)=f(-\lambda).
\]
Note that 
\[
\omega_{\pm}(-k)=-\omega_{\pm}(k)
\]
for all $k\in\mathbb{Z}$. So in the case $\gamma(\beta)>0$, as $k\rightarrow\infty,$ $t=\beta k,$ $\beta>0$,
due to Lemma \ref{FgbetakL} we get: 
\begin{align*}
  Q_{k}(t)&\asymp\frac{1}{\sqrt{|k|}}\bigl(c_{+}(Q(\mu_{+})e^{i\omega_{+}(k)}+Q(-\mu_{+})e^{-i\omega_{+}(k)})\\
  &\quad {} +c_{-}(Q(\mu_{-})e^{i\omega_{-}(k)}+Q(-\mu_{-})e^{-i\omega_{-}(k)})\bigr)
\\
&=\frac{1}{\sqrt{|k|}}\mathcal{F}_{k}^{+}[Q],
\end{align*}
where $\mathcal{F}_{k}^{\pm}$ are defined in Theorem \ref{asympbehgz}.
Similarly one can obtain the expression for $P_{k}(t)$: 
\[
P_{k}(t)=\frac{1}{2i}\bigl(F[g^{*}](\beta,k)-F[g](\beta,-k)\bigr),\quad g=\frac{P(\lambda)}{\omega(\lambda)}.
\]
If $\gamma(\beta)>0$ it is easy to see from the latter formula that
\[
P_{k}(t)\asymp-\frac{i}{\sqrt{|k|}}\mathcal{F}^{-}[g].
\]
This completes the proof of Theorem \ref{asympbehgz}.

\subsection{Proof of Theorem \ref{asympbehezero}}

We will use Lemma \ref{qQPformulaLemma}. The fact that $Q_{k}(t)=O(k^{-n})$
as $t$ is fixed easily follows from the integrating by parts the
corresponding integral $n$ times. Let us consider the term $P_{k}(t)$:
\[
P_{k}(t)=\frac{1}{2\pi}\int_{0}^{2\pi}P(\lambda)\frac{\sin(t\omega(\lambda))}{\omega(\lambda)}e^{-ik\lambda}d\lambda.
\]
In the case $\omega_{0}=0$ we have 
\[
\omega(\lambda)=2\omega_{1}\sin\frac{\lambda}{2}.
\]
Thus the integrand contains two points where the denominator equals
to zero. Since 
\[
\frac{\sin(t\omega(\lambda))}{\omega(\lambda)}=\sum_{k=0}^{\infty}\frac{t^{2k+1}\omega^{2k}(\lambda)}{(2k+1)!},
\]
$[\sin(t\omega(\lambda))] / \omega(\lambda)$ is $C^{\infty}(\mathbb{R})$
smooth function w.r.t.\ $\lambda$ with period $2\pi$. Thus $P_{k}(t)=O(k^{-n})$
as $t$ is fixed. Hence the part 1 of Theorem \ref{asympbehezero}
is proved.

Next we will prove the remainder part of Theorem \ref{asympbehezero}.
We need the following lemma.
\begin{lemma} For all fixed $k\in\mathbb{Z}$
the folowing limit holds 
\begin{equation}
\lim_{t\rightarrow\infty}P_{k}(t)=\frac{P(0)}{2\omega_{1}}. \label{Pklim}
\end{equation}
\end{lemma}
\begin{proof} We need another expression for $P_{k}$: 
\begin{align*}
P_{k}(t)&=\frac{1}{2\pi}\int_{0}^{2\pi}P(\lambda)\frac{\sin(t\omega(\lambda))}{\omega(\lambda)}e^{-ik\lambda}d\lambda
\\
        &=\frac{1}{2\pi}\int_{0}^{2\pi}\frac{P(\lambda)-P(0)}{\omega(\lambda)}\sin(t\omega(\lambda))e^{-ik\lambda}d\lambda\\
  &\quad {} +P(0)\frac{1}{2\pi}\int_{0}^{2\pi}\frac{\sin(t\omega(\lambda))}{\omega(\lambda)}e^{-ik\lambda}d\lambda
\\
&=\tilde{P}_{k}(t)+P(0)I_{k}(t),
\end{align*}
where 
\begin{align*}
  \tilde{P}_{k}(t)&=\frac{1}{2\pi}\int_{0}^{2\pi}\frac{P(\lambda)-P(0)}{\omega(\lambda)}\sin(t\omega(\lambda))e^{-ik\lambda}d\lambda,\\
  I_{k}(t)&=\frac{1}{2\pi}\int_{0}^{2\pi}\frac{\sin(t\omega(\lambda))}{\omega(\lambda)}e^{-ik\lambda}d\lambda .
\end{align*}
Since the integrand $[P(\lambda)-P(0)] / \omega(\lambda)$ in
$\tilde{P}_{k}(t)$ is absolutely integrable, we have due to Riemann\tire Lebesgue
theorem 
\[
\lim_{t\rightarrow\infty}\tilde{P}_{k}(t)=0.
\]
Using lemma \ref{besselExpression} we obtain (\ref{Pklim}).
\end{proof}
\begin{lemma} \label{besselExpression} The following equalities
hold: 
\begin{equation}
I_{k}(t)=\int_{0}^{t}J_{2k}(2\omega_{1}s)ds,\quad\lim_{t\rightarrow\infty}I_{k}(t)=\frac{1}{2\omega_{1}}\label{besselEq}
\end{equation}
where $J_{k}(t)$ is the Bessel function of the first kind. \end{lemma}
\begin{proof} From the definition we have: 
\[
I_{k}(t)=\frac{1}{2\pi}\int_{0}^{2\pi}\frac{\sin(t\omega(\lambda))}{\omega(\lambda)}e^{-ik\lambda}d\lambda=\frac{1}{\pi}\int_{0}^{\pi}\frac{\sin(2\omega_{1}t\sin\lambda)}{2\omega_{1}\sin\lambda}e^{-2ik\lambda}d\lambda.
\]
Thus 
\begin{align*}
\frac{d}{dt}I_{k}(t)&=\frac{1}{\pi}\int_{0}^{\pi}\cos(2\omega_{1}t\sin\lambda)e^{-2ik\lambda}d\lambda
\\
                    &=\frac{1}{\pi}\int_{0}^{\pi}\cos(2\omega_{1}t\sin\lambda)\cos(2k\lambda)d\lambda \\
  &\quad {} -\frac{i}{\pi}\int_{0}^{\pi}\cos(2\omega_{1}t\sin\lambda)\sin(2k\lambda)d\lambda .
\end{align*}
The second term equals zero. To see this one should make the substitution
$x=\pi-\lambda$. Whence 
\begin{align*}
  \frac{d}{dt}I_{k}(t)&=\frac{1}{2\pi}\int_{0}^{\pi}\cos(2\omega_{1}t\sin\lambda+2k\lambda)d\lambda\\
  &\quad {} +\frac{1}{2\pi}\int_{0}^{\pi}\cos(2\omega_{1}t\sin\lambda-2k\lambda)d\lambda
\\
&=\frac{1}{2}(J_{-2k}(2\omega_{1}t)+J_{2k}(2\omega_{1}t))=J_{2k}(2\omega_{1}t).
\end{align*}
It proves the first equality in (\ref{besselEq}). The second one
follows from \cite{GR}, p.\thinspace 1036, 6.511 (1).
\end{proof}

\begin{lemma} \label{ozeroCg} Consider the integral 
\[
C[g](t)=\frac{1}{2\pi}\int_{0}^{2\pi}g(\lambda)\cos(t\omega(\lambda))d\lambda,
\]
where $g(\lambda)\in C^{n}(\mathbb{R}),\ n\geqslant6$ is a complex
valued $2\pi$-periodic function. Then 
\[
C[g](t)=\frac{1}{\sqrt{t}}\frac{g(\pi)}{\sqrt{\pi\omega_{1}}}\cos\Bigl(2\omega_{1}t-\frac{\pi}{4}\Bigr)+b\frac{\cos(2\omega_{1}t- \pi / 4)}{t\sqrt{t}}+O(t^{-2}),
\]
for some complex constant $b\in\mathbb{C}$.
\end{lemma}
\begin{proof}
We apply the stationary phase method. The phase function $\omega(\lambda)=2\omega_{1}\sin(\lambda / 2)$
has the only one critical point $\lambda_{0}=\pi$ on the interval
$[0,2\pi]$. The contribution to the asymptotics of $C[g](t)$ as
$t\rightarrow\infty$ is determined by the boundary points $0,2\pi$
and the critical point $\pi$. The contribution from the boundary
points has the order of $t^{-2}$. Indeed the leading term of the
asymptotic expansion corresponding to the point $0$ is: 
\[
-\frac{1}{2}\Bigl(\frac{g(0)e^{i\omega(0)}}{it\omega'(0)}+\frac{g(0)e^{-i\omega(0)}}{-it\omega'(0)}\Bigr)=0.
\]
Analogously one has for the point $2\pi$. Hence the contribution
to the asymptotics of $C[g](t)$ up to the order $t^{-2}$ determined
by the stationary point. From \cite{Erdelyi} we have the formula
\[
C[g](t)=\frac{1}{2\pi}\frac{\cos(2\omega_{1}t- \pi / 4)}{\sqrt{t}}\Bigl(g(\pi)\sqrt{\frac{4\pi}{\omega_{1}}}+b\frac{1}{t}+O(t^{-2})\Bigr)
\]
for some constant $b$. This completes the proof of Lemma \ref{ozeroCg}.
\end{proof}

Let us continue the proof of Theorem \ref{asympbehezero}. From Lemmas
\ref{qQPformulaLemma} and \ref{ozeroCg} we have: 
\[
Q_{k}(t)\asymp\frac{1}{\sqrt{t}}\frac{(-1)^{k}Q(\pi)}{\sqrt{\pi\omega_{1}}}\cos\Bigl(2\omega_{1}t-\frac{\pi}{4}\Bigr)\ \mathrm{as}\ t\rightarrow\infty.
\]
Using (\ref{Pklim}) we obtain 
\[
\int_{t}^{+\infty}\frac{d}{ds}P_{k}(s)ds=\lim_{T\rightarrow\infty}\int_{t}^{T}\frac{d}{ds}P_{k}(s)ds=\lim_{T\rightarrow\infty}P_{k}(T)-P_{k}(t)=\frac{P(0)}{2\omega_{1}}-P_{k}(t).
\]
Whence due to Lemma \ref{ozeroCg} we get 
\begin{align}
P_{k}(t)&=\frac{P(0)}{2\omega_{1}}-\int_{t}^{+\infty}\frac{d}{ds}P_{k}(s)ds=\frac{P(0)}{2\omega_{1}}-\int_{t}^{+\infty}C[P(\lambda)e^{-ik\lambda}](s)ds
\nonumber \\
        &=\frac{P(0)}{2\omega_{1}}-\int\limits_{t}^{+\infty}\Bigl(\frac{1}{\sqrt{s}}\frac{(-1)^{k}P(\pi)}{\sqrt{\pi\omega_{1}}}
          \cos\Bigl(2\omega_{1}s-\frac{\pi}{4}\Bigr)+b\frac{\cos(2\omega_{1}s- \pi / 4)}{s\sqrt{s}}\Bigr)ds \nonumber \\
  &\qquad {} +O(t^{-1}) . \label{PkexpressionINt}
\end{align}

Applying the second mean-value theorem we have that for all $T\geqslant t$
there is a $\tau\in[t,T]$ such that 
\[
  \int_{t}^{T} \! \frac{\cos(2\omega_{1}s \! - \! \pi / 4)}{s\sqrt{s}}ds=\frac{1}{t\sqrt{t}}\int_{t}^{\tau} \! \cos\Bigl(2\omega_{1}s-\frac{\pi}{4}\Bigr)d\tau+\frac{1}{T\sqrt{T}}
  \int_{\tau}^{T} \! \cos\Bigl(2\omega_{1}s-\frac{\pi}{4}\Bigr)ds.
\]
It is clear that 
\[
\sup_{-\infty<a<b<+\infty}\biggl| \int_{a}^{b}\cos\Bigl(2\omega_{1}s-\frac{\pi}{4}\Bigr)ds\biggr| =\frac{1}{\omega_{1}} .
\]
Therefore 
\[
\biggl| \int_{t}^{T}\frac{\cos(2\omega_{1}s- \pi / 4)}{s\sqrt{s}}ds\biggr| \leqslant\frac{2}{\omega_{1}}\frac{1}{t\sqrt{t}}.
\]
So we obtain 
\begin{equation}
\biggl| \int_{t}^{+\infty}\frac{\cos\Bigl(2\omega_{1}s- \pi / 4 \Bigr)}{s\sqrt{s}}ds\biggr|=O(t^{-3/2}).\label{costhreehalfestim}
\end{equation}
Integrating by parts we have 
\begin{align*}
  &\int_{t}^{+\infty}\frac{1}{\sqrt{s}}\cos\Bigl(2\omega_{1}s-\frac{\pi}{4}\Bigr)ds\\
  &\quad =-\frac{1}{2\omega_{1}\sqrt{t}}\sin\Bigl(2\omega_{1}t-\frac{\pi}{4}\Bigr)+\frac{1}{4\omega_{1}}\int_{t}^{+\infty}\frac{1}{s\sqrt{s}}
    \sin\Bigl(2\omega_{1}s-\frac{\pi}{4}\Bigr)ds=
\\
&\quad =-\frac{1}{2\omega_{1}\sqrt{t}}\sin\Bigl(2\omega_{1}t-\frac{\pi}{4}\Bigr)+O(t^{-3/2}).
\end{align*}
In the last equality we have used the same estimate for the remainder
term as in (\ref{costhreehalfestim}), which can be easily proved
in the same way.

Summing up obtained estimates for (\ref{PkexpressionINt}) we get
\[
P_{k}(t)=\frac{P(0)}{2\omega_{1}}+\frac{(-1)^{k}P(\pi)}{2\omega_{1}\sqrt{\pi\omega_{1}t}}\sin\Bigl(2\omega_{1}t-\frac{\pi}{4}\Bigr)+O(t^{-3/2}).
\]
This completes the proof of Theorem \ref{asympbehezero}.

\subsection{Acknowledgment}

We would like to thank professor Vadim Malyshev for stimulating discussions
and numerous remarks.

\end{document}